\documentclass[11pt,english,letterpaper]{article}

\usepackage[margin=1in,footskip=0.25in]{geometry}
\usepackage{amssymb}
\usepackage{graphicx}
\usepackage{microtype}
\usepackage{amsmath}
\usepackage{amsthm}
\usepackage{enumerate}
\usepackage{enumitem}
\usepackage{wrapfig}
\usepackage{multirow}
\usepackage{nicefrac}
\usepackage[colorinlistoftodos,prependcaption,textsize=small]{todonotes}
\usepackage[colorlinks,linkcolor=blue,filecolor=blue,citecolor=blue,urlcolor=blue,pdfstartview=FitH,pagebackref]{hyperref}
\usepackage[nameinlink]{cleveref}
\usepackage{authblk}
\usepackage{multicol}
\usepackage{thm-restate}
\usepackage{microtype}
\usepackage{url}
\usepackage{wrapfig}
\usepackage{multirow}
\usepackage{subfigure}
\usepackage{makecell}

\usepackage{algorithm}
\usepackage[noend]{algpseudocode}

\definecolor{forestgreen}{rgb}{0.13, 0.55, 0.13}

\newtheorem{theorem}{Theorem}
\newtheorem{lemma}{Lemma}
\numberwithin{lemma}{section}
\newtheorem{claim}[lemma]{Claim}

\newtheorem{observation}[lemma]{Observation}

\newtheorem{remark}[lemma]{Remark}

\newtheorem{problem}{Problem}

\graphicspath{{./figures/}}

\def\A{\mathcal{A}}

\def\R{\mathcal{R}}

\def\U{\mathcal{U}}
\def\reals{{\mathbb R}}

\def\eps{{\varepsilon}}

\def\VP{\text{Vis}}        
\def\OPT{\text{OPT}}

\def\dfn#1{\emph{\textcolor{forestgreen}{\textbf{#1}}}}

\date{}

\newcommand{\old}[1]{{{}}}
\newcommand{\temp}[1]{{{}}}
\newcommand{\fullversion}[1]{{{}}}

\newcommand{\omrit}[1]{}
\newcommand{\omritin}[1]{}
\newcommand{\ofir}[1]{}
\newcommand{\ofirin}[1]{}

\title{Peeling Rotten Potatoes for a Faster Approximation of Convex Cover}

\author{Omrit Filtser}
\author{Tzalik Maimon}
\author{Ofir Yomtovyan}

\affil{\small Department of Mathematics and Computer Science, The Open University of Israel\\
\texttt{omrit.filtser@gmail.com}, \texttt{post.tmx@gmail.com}, \texttt{ofyotb1@gmail.com}}

\date{}

\begin{document}

\maketitle

\begin{abstract} 
The minimum convex cover problem seeks to cover a polygon $P$ with the fewest convex polygons that lie within $P$. This problem is $\exists\reals$-complete, and the best previously known algorithm, due to Eidenbenz and Widmayer (2001), achieves an $O(\log n)$-approximation in $O(n^{29} \log n)$ time, where $n$ is the complexity of $P$.

In this work we present a novel approach that preserves the $O(\log n)$ approximation guarantee while significantly reducing the running time. By discretizing the problem and formulating it as a set cover problem, we focus on efficiently finding a convex polygon that covers the largest number of uncovered regions, in each iteration of the greedy algorithm. 
This core subproblem, which we call the \emph{rotten potato peeling} problem, is a variant of the classic potato peeling problem. We solve it by finding maximum weighted paths in Directed Acyclic Graphs (DAGs) that correspond to visibility polygons, with the DAG construction carefully constrained to manage complexity.
Our approach yields a substantial improvement in the overall running time and introduces techniques that may be of independent interest for other geometric covering problems.
\end{abstract}

\section{Introduction}\label{sec:intro}
The \textit{Minimum Convex Cover} (MCC) problem is a fundamental problem in computational geometry: cover a given polygon $P$ (possibly with holes) using the fewest number of (possibly overlapping) convex polygons that lie within $P$.
The MCC problem is known to be NP-hard, even for simple polygons~\cite{CulbersonReckhow88, ORourkeSupowit83}, and recent work has shown it to be $\exists\mathbb{R}$-complete \cite{Abrahamsen21}. This implies that optimal solutions may require irrational coordinates, ruling out exact formulations via integer or rational arithmetic. Even when covering simple polygons with triangles, the problem remains $\exists\mathbb{R}$-hard.

The MCC problem has garnered significant attention in recent years, including being the focus of the CG:SHOP Challenge 2023~\cite{fekete2023minimum} --- a major international competition that attracted 22 teams and highlighted the practical importance and theoretical difficulty of the problem.
Despite this attention, there has been no improvement on the state of the art $O(\log n)$-approximation algorithm found by Eidenbenz and Widmayer~\cite{EidenbenzWidmayer03} in over two decades, largely due to the prohibitive running time of existing methods.

Eidenbenz and Widmayer \cite{EidenbenzWidmayer03} provided the first non-trivial approximation algorithm for general polygons. They achieve an $O(\log n)$-approximation factor by showing that solving a restricted (discrete) version of convex cover yields a constant-factor approximation to MCC. This restricted version can be formulated as a set cover problem, and by applying the classic greedy algorithm obtain an $O(\log n)$-approximation algorithm for MCC. Each iteration of the greedy algorithm requires finding a convex polygon that covers the largest number of uncovered regions in a partition of $P$. However, their approach suffers from an extremely high running time of $O(n^{29} \log n)$, because it includes the construction of a dense quasi-grid with $O(n^{16})$ triangles, and then iteratively selecting convex polygons using dynamic programming over this set, taking $O(n^{28})$ time per iteration.

For simple polygons, very recently, Browne, Kasthurirangan, Mitchell, and Polishchuk~\cite{BKP&M} presented a $6$-approximation algorithm that runs in $O(n^2)$ time. Their method employs window partitioning of the input polygon, which guarantees that no convex polygon in the optimal cover can intersect more than three regions of the partition. This structural property is critical to their analysis. However, the approach does not extend to polygons with holes, as the presence of holes can cause a convex polygon to intersect multiple disjoint regions of the partition, breaking the underlying assumption and invalidating the approximation guarantee.

In this paper we use a general approach similar to Eidenbenz and Widmayer by formulating the problem as a set cover problem and running the greedy algorithm. However, we show that it is enough to construct a much smaller set of regions for the partition of $P$, and describe a novel DAG-based approach to find large convex polygons in visibility polygons, that leads to a significant improvement in the running time, which then becomes $O(n^8)$.

Our DAG approach in visibility polygons was inspired by the recent work of Browne et al.~\cite{BKP&M}, although it is very different and solves a completely different problem. 
In~\cite{BKP&M}, the polygon $P$ is partitioned into weakly visible polygons (where there is a single edge that see the entire polygon), and in each such polygon they construct a DAG on its edges, in which a minimum path cover corresponds to a minimum convex cover for its boundary. Our DAGs are defined on a partition of visibility polygons (and not just its edges), and, roughly speaking, we show that a maximum (weighted) path in such a DAG corresponds to a convex polygon that covers the maximum number of good regions.

In addition, we notice that each iteration of the greedy algorithm solves a variant of the well-known potato peeling problem. The classic potato peeling problem is defined as finding the largest (area) convex polygon contained in a polygon $P$. In our case, we wish to cover only regions of $P$ that were not covered yet, however, we do not want to regard them as holes of the polygon, because the convex polygons in the solution may overlap. We therefore regard those regions as ``rotten'' and we wish to find the convex polygon which covers the largest amount (by area or number) of regions which are non-rotten. 

The potato-peeling problem (also known as convex skull) is a fundamental problem in computational geometry, introduced in an early work by Goodman~\cite{goodman1981largest} in 1981.
The best known (exact) algorithm for simple polygons runs in $O(n^7)$ time, and in $O(n^8)$ for polygons with holes \cite{chang1986polynomial}. There is also a $(1-\eps)$-approximation randomized algorithm that runs in $O\left(n\left(\log^2 n + \frac{\log n}{\varepsilon^3} + \frac{1}{\varepsilon^4}\right)\right)$ time, with success probability $\geq 2/3$.

Several different variants of the potato peeling problem have been considered in the literature. Hall‑Holt et al.~\cite{hall2006finding} designed an $O(n\log n)$-time constant factor approximation algorithm for computing largest simple shapes (such as triangles and ellipses) inside a polygon. Crombez, da Fonseca and Gérard \cite{crombez2018peeling} consider a discretized version of the problem on lattice points. Aronov et al.\cite{aronov2011peeling} consider the problem of peeling meshed potatoes, where the polygon $P$ is given with a triangulation (possibly with additional interior vertices). They present an exact $O(m^2)$ algorithm to find the convex subpolygon of maximum area that is a union of triangles from the mesh, where $m$ is the number of triangles in the mesh. 

\vspace{-7pt}
\paragraph{Our contribution.}
We present an $O(n^8)$-time $O(\log n)$-approximation algorithm for the minimum convex cover problem. While still being a high polynomial running time, our result
substantially improves the running time of the state of the art $O(\log n)$-approximation algorithm —-- from $O(n^{29} \log n)$ to $O(n^8)$  --— making the first significant progress on this longstanding open problem in over two decades.

Our algorithm is based on a formulation of the MCC problem as a set cover instance, where each iteration involves solving a generalization of the classic potato peeling problem. 
Given a polygon $P$ and a set $R$ of rotten regions in $P$, the goal is to find a convex polygon $Q \subseteq P$ that maximizes some weighted measure (such as area or number of subregions from a given partition of $P$), where only $Q\setminus R$ contributes to the objective.
We define two variants of this \emph{rotten potato problem}: one maximizing the weight of $Q\setminus R$, and the other maximizing the number of subregions in $Q\setminus R$ from a given partition of $P$ (meshed potato peeling). The general approach in our solution for both problems is similar: via a reduction to the longest-path problem, in a carefully constructed geometric DAG. The second version is mainly used as a sub-procedure for solving the MCC problem. For the first version, which is more natural, we obtain a $\frac{1}{4}$-approximation in $O(kn^5)$ time, where $k$ is the complexity of $R$. In the special case with no rotten regions, our method recovers a $\frac{1}{4}$-approximation for the classical potato peeling problem in $O(n^5)$ time.

\section{Outline of the algorithm}\label{sec:outline}
Let $P$ be a polygon (possibly with holes), and denote by $V$ the set of its $n$ vertices. Denote by $\partial P$ the boundary of $P$. In this paper, whenever we say that we traverse the boundary of $P$, we do so in counterclockwise (CCW) order.

Let $\OPT$ be a minimum convex cover for $P$, i.e. $\OPT$ is a minimum size set of convex polygons in $P$, such that $\bigcup_{Q\in \OPT}Q=P$. Note that $|\OPT|=O(n)$, because any triangulation of $P$ covers $P$ and has size $n+2h-2$, where $h$ is the number of holes in $P$.

\paragraph{A discretization of the polygon.}
For two points $u,v$ such that $\overline{uv}\subset P$, the \dfn{extension} of $\overline{uv}$ is the maximal segment $\overline{ab}\subseteq P$ such that $\overline{uv}\subseteq \overline{ab}$. If $u$ and $v$ are vertices of $P$, then the extension of $\overline{uv}$ is called a \dfn{diagonal extension}. Denote by $D$ the set of diagonal extensions, and note that $|D|=O(n^2)$. Denote by $V_D$ the set of intersection points between diagonal extensions from $D$, and note that $V_D$ contains $V$, and that $|V_D|=O(n^4)$. Let $\A_D$ be the arrangement of segments from $D$, and denote by $\U$ the set of faces in $\A_D$.

\paragraph{Restricted convex cover.} 
We say that a polygon $Q\subseteq P$ is \dfn{restricted} if it is \textbf{convex}, its vertices are from $V_D$ and its edges are subsegments of the diagonal extensions in $D$. Denote by $\R$ the set of all restricted polygons.
Observe that any restricted polygon is a connected union of a subset of faces from $\U$.
In \Cref{sec:restricted_convex_cover}, we prove that for any convex polygon $Q$ in $P$, there exist at most three restricted polygons that cover it (see \Cref{thm:approx}). This means that there exists a subset $S^*\subseteq\R$ of restricted polygons such that $S^*$ covers $P$ (i.e. $\bigcup_{Q\in S^*} Q=P$), and $S^*\le 3|\OPT|$.
We therefore define the \dfn{restricted convex cover} (RCC) problem as follows. Given a polygon $P$, find a set of restricted polygons that together cover $P$. By \Cref{thm:approx}, an optimal solution for RCC is a $3$-approximation for MCC. We can now formulate RCC as set cover problem, where the set of elements is $\U$, and the collection of sets is $\R$.

\paragraph{Peeling rotten potatoes for greedy set cover.} To get an $O(\log n)$-approximation algorithm for RCC (and by that, a $O(\log n)$ approximation to MCC), we run the greedy algorithm that repeatedly chooses a restricted polygon that covers the largest number of faces from $\U$ that were not covered yet. It is well-known that this greedy algorithm returns a $O(\log |\U|)=O(\log n)$ approximation. 
At each greedy step, we need to solve the following problem. 
Given the set of faces $\U$, with a binary function $w:\U\rightarrow\{0,1\}$, find a restricted polygon $Q\in\R$ that maximizes $w(Q)=\sum_{f\in \U, f\subseteq Q} w(f)$. In other words, find the restricted polygon that covers the largest number of faces $f\in \U$ with $w(f)=1$ (which would mean that they were not covered yet). 
We refer to this problem as the \dfn{rotten meshed potato peeling} problem.
In \Cref{sec:potato} we present an algorithm that given $P$, $\U$ and the function $w$, finds in $O(n^7)$ time at most two restricted polygons $Q_1,Q_2\in\R$ such that $w(Q_1\cup Q_2)\ge w(Q^*)$, where $Q^*$ is the restricted polygon that maximizes $w(Q^*)$.
We then use this algorithm as a sub-procedure in the greedy algorithm.
If the algorithm did not terminate before the $n$-th round, we can stop and simply return a triangulation of $P$ with $n+2h-2=O(n)$ triangles. Therefore, the total running time of our $O(\log n)$-approximation algorithm is $O(n^8)$.

The rotten meshed potato peeling problem may sound ad hoc; however, the algorithm that we develop, along with the geometric insights behind its construction, may be of independent interest. In fact, as stated in the introduction, in \Cref{sec:more-potatoes} we also present a more natural variant of the rotten potato peeling problem, which we solve using similar techniques.

\section{Restricted convex cover} \label{sec:restricted_convex_cover}
Our goal in this section is to show that the set $\R$ of restricted polygons contains a constant approximation for MCC.
As mentioned in the introduction, Eidenbenz and Widmayer~\cite{EidenbenzWidmayer03} define restricted polygons as polygons whose vertices are from $V_D$, but their edges may not be subsegments of diagonal extensions. However, we observe that by a slight modification to their proof of Theorem~1 in their paper, we get that any convex polygon in $P$ can be replaced by at most three restricted polygons, as in our definition. Since their proof spreads on over more than three pages and it mostly remains unchanged, we will only give a sketch of the similar parts of the proof, and detailed arguments only in the modified parts. 

Note that we include an additional (new) claim in the statement of the theorem, showing that if the given convex polygon contains a convex vertex of $P$, then it can be replaced by only \textbf{two} restricted polygons. This addition will be useful later in \Cref{sec:more-potatoes}, where we present a $1/4$-approximation algorithm for a more natural version of the rotten potato peeling problem.

\begin{theorem}[based on Theorem 1 from \cite{EidenbenzWidmayer03}]\label{thm:approx}
    Let $Q$ be a maximal convex polygon contained in a $P$. Then, there always exists at most three restricted polygons that cover $Q$. Moreover, if $Q$ contains a convex vertex of $P$, then it can be covered by at most two restricted polygons.
\end{theorem}
\begin{proof}[Proof]
    The proof starts by showing how to expand a given convex polygon $Q$ to another convex polygon that contains it, by iteratively expanding the edges of $Q$. 
    We call a vertex of $Q$ a $P$-vertex if it is also a vertex of $P$. An edge $\overline{ab}$ of $Q$ is \emph{expandable} if both $a,b$ are not $P$-vertices, and do not lie on the same edge of $P$. An expandable edge can be ``pushed'' (parallel to itself) outwards from $Q$ while remaining between the extensions of its neighbor edges, until it touches a vertex or an edges of $P$ and stops being expandable (or until it becomes a vertex - the intersection point of the two extensions).
    
    After expanding all the expandable edges, we get a polygon $Q'$ with the property that there are no three consecutive vertices of $Q'$ that are not $P$-vertices, and if there are two consecutive non-$P$-vertices, then they must lie on the same edge of $P$.
    
    \begin{figure}[h!]
        \centering
        \includegraphics[scale=1]{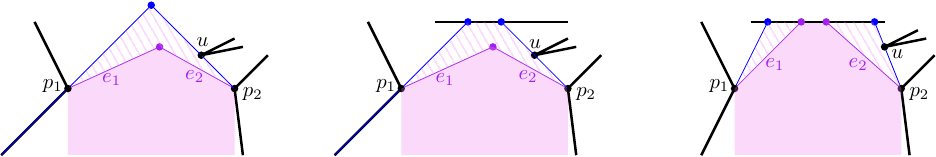}
        \caption{Constructing a non-convex polygon $Q''$ that contains $Q'$.}
        \label{fig:restricted_approx}
    \end{figure}
    The next step is to construct a non-convex polygon $Q''$ that contains $Q'$, and in which all the edges lie on diagonal extensions of $P$. This is done by separately expanding each subchain of the boundary of $Q'$ between two consecutive $P$-vertices. This is the only step where we need a slightly different argument than the one in \cite{EidenbenzWidmayer03}. Let $p_1$ and $p_2$ be two consecutive $P$-vertices of $Q$ (see \Cref{fig:restricted_approx}). 
    Then between $p_1$ and $p_2$ there is either a single non-$P$-vertex, or two non-$P$-vertices that lie on the same edge of $P$ (or no vertex at all, in which case we are done). 
    Let $e_1$ be the edge from $p_1$ to the next vertex of $Q$, and $e_2$ the edge from $p_2$ to the vertex of $Q$ before $p_2$. 
    The subchain of the boundary of $Q'$ between $p_1$ and $p_2$ splits $P$ into two subpolygons, where $Q'$ is contained in one of them. Let $e'_1$ (resp. $e'_2$) be the edge of $P$ incident to $p_1$ (resp. $p_2$) in the subpolygon that contains $Q'$.
    Rotate the ray from $p_1$ through $e_1$ in the direction outwards from $Q$, until it either touches another vertex $u$ of $P$ (and thus aligned with a diagonal extension of $\overline{pu}$), or it is aligned with the diagonal extension $e'_1$.
    Do the same for $p_2$.
    There are three cases as depicted in the figure above. 
    \begin{enumerate}[noitemsep]
        \item There is a single non $P$-vertex between $p_1$ and $p_2$, and the diagonal extensions intersect.
        \item There is a single non $P$-vertex between $p_1$ and $p_2$, and the diagonal extensions do not intersect. In this case they must reach the same edge of $P$, as otherwise the rays would stop rotating when reaching that vertex.
        \item There are two non-$P$-vertices that lie on the same edge $e$ of $P$. In this case both rays must intersect $e$ as well, because they stop rotating when reaching a vertex of that edge.
    \end{enumerate}
    In each case, we can extend the subchain by replacing it with the intersection point of the diagonal extensions in case 1, or by the intersection points of the diagonal extension with an edge of $P$ (which also lies on a diagonal extension) in cases 2 and 3. In all cases, the edges of $Q$ between $p_1$ and $p_2$ lie on diagonal extensions, and the vertices are from $V_D$.
    
    The last step is exactly as in \cite{EidenbenzWidmayer03}: construct at most three convex polygons that contain $Q''$ (see \Cref{fig:restricted_approx2}). Denote by $\hat{Q}$ the convex polygon whose vertices are all the $P$-vertices of $Q''$. Two polygons, $Q_{\text{even}}$ and $Q_{\text{odd}}$, are constructed as follows. $Q_{\text{even}}$ is constructed by adding to $\hat{Q}$ the even subchains between $P$-vertices, and $Q_{\text{odd}}$ is constructed by adding to $\hat{Q}$ the odd subchains (see the figure below). Notice that there must be an even number of subchains in each of the polygons $Q_{\text{even}}$ and $Q_{\text{odd}}$, because connecting two consecutive $P$-vertices may result in a non-convex polygon. Therefore, if there is an odd number of subchains (i.e., an odd number of vertices in $\hat{Q}$), then we add another polygon that contains a single subchain (between two consecutive $P$-vertices of $Q$). The edges of these polygons are either diagonals of $P$, or edges of $Q''$, so all three polygons are restricted polygons (by our definition). 

    \begin{figure}[h!]
        \centering
        \includegraphics[scale=1]{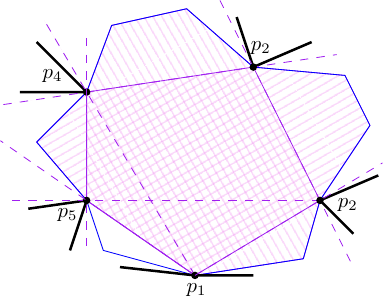}
        \caption{Constructing at most three convex polygons that contain $Q''$.}
        \label{fig:restricted_approx2}
    \end{figure}
    To prove the second part of the theorem, observe that $Q''$ can only have reflex vertices at reflex vertices of $P$, because the angles of the vertices on each subchain between two consecutive $P$-vertices are convex. If $Q''$ has a convex $P$-vertex $p_1$ (as in \Cref{fig:restricted_approx2}), then the angle at $p_1$ (that connects two such consecutive subchains) is also convex. Therefore, in this case, we can construct $Q_{\text{even}}$ and $Q_{\text{odd}}$ as above, but such that the remaining third polygon includes $p_1$. Then, since the angle to one of its neighbor subchains is convex, we can add it to the polygon that covers this subchain, and thus cover $Q''$ with at most two convex polygons.
\end{proof}

We conclude with the following theorem.
\begin{theorem}
    There exists a subset $S^*\subseteq\R$ of restricted polygons such that $S^*$ covers $P$ (i.e. $\bigcup_{Q\in S^*} Q=P$), and $S^*\le 3|\OPT|$.
\end{theorem}

\section{Peeling rotten potatoes for greedy set cover}\label{sec:potato}

Our main goal in this section is to solve the following problem.
\begin{problem}[Rotten Meshed Potato Peeling]\label{prb:meshed-potatoes}
    Given a polygon $P$ with $n$ vertices, the set of faces $\U$ (as defined in \Cref{sec:outline}), and a function $w:\U\rightarrow \{0,1\}$, find a restricted polygon $Q$ in $P$ that maximizes $\sum_{f \in \U, f \subseteq Q} w(f)$.
\end{problem}
To simplify the presentation, in this section we call such a restricted polygon $Q$ a \dfn{largest restricted polygon}.

As mentioned in \Cref{sec:outline}, in order for the solution to be utilized for the greedy MCC algorithm, we present an algorithm that gets the function $w$ as an input, and computes the largest restricted polygon. The function $w$ is then updated accordingly. Our algorithm is based on constructing a weighted DAG for each vertex of $P$, and then finding a heaviest path in the DAG.

\paragraph{Peeling visibility polygons.}
For a point $x\in P$, the visibility polygon of $x$ is defined as $\VP(x)=\{p\in P\mid \overline{xp}\subseteq P\}$.
A key insight in our construction is that it is sufficient to search for the largest restricted polygon within the visibility polygons of the vertices of $P$. 

This is because in the proof of \Cref{thm:approx}, the construction of restricted polygons ensures that each such polygon contains at least one vertex of $P$.
In other words, there exists a set of $3|OPT|$ restricted polygons that cover $P$, and each of them contains a vertex of $P$. We can therefore further restrict our set $\R$ to polygons that contain a vertex of $P$. 
Since restricted polygons are convex, if a restricted polygon $Q$ contains a vertex $v$ of $P$, then $Q$ must lie entirely within the visibility polygon of $v$.

\begin{observation}
    Let $Q$ be a restricted polygon that contains a vertex $v$ of $P$, then $Q\subseteq\VP(v)$.
\end{observation}
Therefore, it is enough to iterate over all the vertices of $P$, and find for each vertex $v$ the largest restricted polygon in $\VP(v)$. Then we can take the maximum over all the restricted polygons that were found.

\subsection{DAG construction for convex vertices} \label{sec:dag_const}
Fix a vertex $v$ of $P$, and consider its visibility polygon $P_v=\VP(v)$. We assume in this section that $v$ is a convex vertex (i.e., the internal angle at $v$ is at most $180^\circ$).

Let $D_v=\{l\cap P_v\}_{l\in D}$ be the set of diagonal extensions restricted to $P_v$ (see \Cref{fig:SSTs}). The segments in $D_v$ are subsegments of the segments in $D$, and are maximal w.r.t. $P_v$. In other words, each segment in $D_v$ has both its endpoints on $\partial P_v$ (the boundary of $P_v$). 

\paragraph{Directed subsegments.} We now assign directions to the segments in $D_v$, as follows. For any segment $\overline{ab}\subseteq P_v$ with both endpoints $a,b$ on $\partial P_v$, such that $a$ appears before $b$ when traversing the boundary in counterclockwise order from $v$, and $a\neq b\neq v$, the direction is from $a$ to $b$, and we denote the directed segment by $\overrightarrow{ab}$. If, e.g., $a=v$, then the direction is $\overrightarrow{vb}$.
For each directed segment $l=\overrightarrow{ab}\in D_v$, let $a=q_1,q_2,\dots,q_t=b$ be the sequence of points from $V_D$ on $l$, by their order along $l$. Every two consecutive points $q_i,q_{i+1}$ define a \dfn{directed subsegment} of $l$.

\paragraph{Weighted cell Partition.} Denote by $B_v$ the set of extensions of segments $\overline{vp}$ such that $p\in V_D\cap P_v$ (see \Cref{fig:SSTs}).
That is, $B_v$ is the set of extensions of all the segments connecting $v$ with a point from $V_D$ in its visibility polygon. We call such an extension a \dfn{visibility extension}. Notice that $|B_v|=O(n^4)$ because $|V_D|=O(n^4)$.
Consider the set of cells $C_v$ in $P_v$ formed by the segments in $D_v\cup B_v$.
This set of cells is a refinement of the arrangement $\A_D$, and we refer to it as the \dfn{cell partition} of $P_v$. Each cell $c \in C_v$ is contained in some face $f\in\U$. For a face $f\in\U$ that contains a set $C'\subseteq C_v$ of cells, we set the weight of each cell $c\in C'$ to be $w(c)=\frac{w(f)}{|C'|}$. That is, if $w(f)=0$ then the weight of each cell $c$ in the face $f$ is $0$, and otherwise it is one over the number of cells in $f$.

Notice that the segments in $B_v$ have a common intersection point (at $v$), and each of them may intersect each of the segments in $D_v$. Because $|B_v|=O(n^4)$ and $|D_v|=O(n^2)$, the total number of intersection points of segments in $D_v\cup B_v$ is $O(n^6)$, and thus we have $|C_v| = O(n^6)$.

\paragraph{Subsegment triangles.}
For a directed subsegment $s=\overrightarrow{ab}$ of a segment in $D_v$, let $T_v(s)=\triangle vab$ be the \dfn{subsegment triangle} (SST) of $s$ in $P_v$. Notice that by definition, the segments $\overline{va}$ and $\overline{vb}$ are lying on visibility extensions from $B_v$, because $a,b\in V_D$. 
Therefore, any such SST is a connected subset of cells from $C_v$ (i.e., they are connected in the corresponding dual graph). Denote by $C_v(s)$ the set of cells from $C_v$ that are contained in the SST $T_v(s)$. We set the weight of $T_v(s)$ to be $w(T_v(s))=\sum_{c \in C_v(s)} w(c)$.

    \begin{figure}[h]
        \centering
        \includegraphics[page=1, scale=1.2]{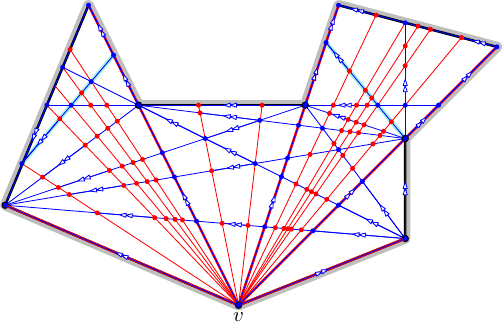}
        \caption{A visibility polygon $P_v$ with the diagonal extensions in $D_v$ (blue), and visibility extensions $B_v$ (red). Two diagonal extensions are highlighted in turquoise, to show that diagonal extensions are coming from the entire polygon $P$. The blue points are the points of $V_D$, and the red points are the points of $W$. Each SST is a connected subset of cells from $C_v$.}
        \label{fig:SSTs}
    \end{figure}

\paragraph{} Using the above definitions, we construct a weighted directed acyclic graph $G=(V(G),E(G),w')$ as follows. %
The set of vertices of $G$ is the set of all directed subsegments in $P_v$, i.e. $$V(G)=\{s\mid s \text{ is a directed subsegment of some } l\in D_v\}.$$
For a segment $s\in V(G)$, we let $w'(s)=w(T_v(s))$.
There is a directed edge $e=(s_1,s_2)$ in $E(G)$ between two directed subsegments $s_1,s_2\in V(G)$ if $s_1 = \overrightarrow{ab}$, $s_2=\overrightarrow{bc}$, and $c$ lies on the left side of the directed line through $\overrightarrow{ab}$. In other words, when walking from $a$ to $b$ and then to $c$ we make a left-turn --- the angle $\angle abc$ is at most $180^\circ$. 
In the rest of the paper, whenever we refer to the angle $\angle abc$ between two directed subsegments $\overrightarrow{ab}$, $\overrightarrow{bc}$, we mean the left-turn angle as defined here.

To show that $G$ is a DAG, we first need the following observation.
\begin{observation}\label{obs:order}
    Let $P$ be a simple polygon, and $s_1=\overline{ab}$, $s_2=\overline{cd}$ two maximal segments in $P$ (i.e, their endpoints $a,b,c,d$ are on $\partial P$). If $s_1$ and $s_2$ cross, then their counterclockwise order along the boundary of $P$ is either $a,c,b,d$ or $a,d,b,c$.
\end{observation}
\begin{proof}
    The segment $s_1=\overline{ab}$ splits the polygon into two sub polygons, and in order for the segments $s_1$ and $s_2$ to cross, the points $c,d$ have to lie in different subpolygons. Therefore, when traversing $\partial P$ in counterclockwise order from $a$, we first meet either $c$ or $d$, and then $b$.
\end{proof}

\begin{claim}
    The graph $G$ defined above is a DAG.
\end{claim}
\begin{proof}
    Assume by contradiction that $G$ is not acyclic, so there exists a cycle $\pi=\{s_1,s_2,\dots,s_k\} \in G$. First notice that the vertex $v$ cannot participate in $\pi$, because it has no incoming edges.
    Denote the endpoints of the directed segment that defines the subsegment $s_i\in \pi$ by $\overrightarrow{a_ib_i}$. By the definition of $G$, every two consecutive directed segments $\overrightarrow{a_ib_i}$, $\overrightarrow{a_{i+1}b_{i+1}}$ cross each other.
    Since $P_v=\VP(v)$ is a simple polygon, by \Cref{obs:order} we get that their counterclockwise order along $\partial P_v$ is $a_i, a_{i+1}, b_i, b_{i+1}$ ($a_i, b_{i+1}, b_i, a_{i+1}$ is not possible by the definition of $E(G)$, because then the direction of $s_i,s_{i+1}$ would be toward their intersection point). Assume w.l.o.g. that $a_2$ is the first endpoint on $\partial P_v$ after $v$. Then the order along $\partial P_v$ is $a_1, v, a_2, b_1, b_2$, but this is in contradiction to the direction of $\overrightarrow{a_1b_1}$.
\end{proof}

\begin{remark}
    Note that in the above claim we did not use the fact that $P_v$ is a visibility polygon, and interestingly, in fact $G$ is a DAG for any simple polygon and a set of maximal segments.
\end{remark}

\paragraph{} Next, we show how to construct $G$ for the vertex $v$, given $D$ and the arrangement $\U$.
First notice that any two diagonal extensions that intersect are responsible for exactly four edges in $E(G)$, and thus $|E(G)|=O(n^4)$. Moreover, we can construct $V(G)$ and $E(G)$ in $O(n^4)$ time by using standard data structures for representing planar maps, such as DCEL \cite{de2008computational}. 

The set $B_v$ of visibility extensions can be computed in $O(n^4)$ time. Then, constructing the weighted cell partition $C_v$ can be done in $O(n^6)$ time by sweeping the faces of the arrangement $\U$ according to the segments in $B_v$. During the sweeping process we also keep pointers between cells in $C_v$ and their corresponding faces from $\U$. The weights of the cells in $C_v$ can be computed afterwards by simply traversing the faces in $\U$. After each iteration of the greedy MCC algorithm (in which we add a restricted polygon $Q$ to the cover), we need to update the weights of the cells in $C_v$. Recall that restricted polygons are connected unions of faces from $\U$. Whenever a restricted polygon $Q$ is selected for the cover, we update an indicator for each of these faces, recording whether it has already been covered or not. This step takes $O(n^4)$ time, and then we can update the weight of each cell $c\in C_v$ in constant time by simply checking if its corresponding face from $\U$ was covered. Notice that in this case, the weight $w(c)$ becomes $0$ (during the entire algorithm, the value of $w(c)$ changes only once).

Computing the weight function $w'$ is slightly more complicated. There are $O(n^4)$ SSTs, each consists of a set of $O(n^{6})$ cells from $C_v$, so a trivial approach would lead to a running time of $O(n^{10})$. 
In \Cref{lem:constructTime} below, we show that given the weights of all cells $c\in C_v$ we can compute the the weights of the SSTs in $O(n^6)$ time. 
This is done by computing the weight function for some auxiliary triangles, each is a connected union of cells from $C_v$.

\begin{lemma}\label{lem:constructTime}
    Given a visibility polygon $P_v$ and the cell partition $C_v$ with a weight $w(c)$ for every cell $c\in C_v$, the weights of all the SSTs in $P_v$ can be computed in $O(n^{6})$ time in total. 
\end{lemma}
\begin{proof}
    We define a set of auxiliary subsegment triangles, which is a refinement of the SSTs, as follows. Let $W$ be the set of vertices of the arrangement $C_v$ that are not in $V_D$ (see \Cref{fig:SSTs}). Each such point corresponds to an intersection point between a segment $l_1 \in D_v$ with a segment $l_2 \in B_v$ (by definition, the segments in $B_v$ intersect only at $v$). The points in $W$ subdivide the directed subsegments of each directed segment in $D_v$ into auxiliary subsegment. Consider a directed segment $l=\overline{ab}\in D_v$, and denote the sequence of points from $V_D\cup W$ on $l$ by $a=q_1,q_2,\dots,q_k=b$. We add an auxiliary subsegment triangle (AST) $\triangle vq_iq_{i+1}$ for each auxiliary subsegment. Each such AST is a connected union of cells from $C_v$.
    
    Consider the set of ASTs in the cone between two consecutive visibility extensions $l_1,l_2\in B_v$. The auxiliary subsegments defining these ASTs are non-crossing (otherwise there would be a point from $V_D$ between them, for which we add another visibility extension), and therefore the weights of the ASTs in this cone can be computed in time linear in their number, by accumulating the weights of cells in $C_v$ in the order of containment between the ASTs.

    After we have computed the weights of all ASTs, we can compute the weights of the SSTs as follows. Consider a directed segment $l=\overline{ab}\in D_v$ with a sequence of points $a=q_1,q_2,\dots,q_k=b$ from $V_D\cup W$ as before. Notice that $q_1\in V_D$ because $l$ is maximal. Starting from $q_3$ (if $k=2$, we are done), we iterate over the sequence of points in their order along $l$, and update the weights of their corresponding ASTs as follows: if $q_{i-1}\in W$, set $w(\triangle vq_{i-1}q_i)=w(\triangle vq_{i-1}q_i)+w(\triangle vq_{i-2}q_{i-1})$. In other words, we accumulate the weights of ASTs along $l$ until we get to a point from $V_D$. The result of this process is that for a point $q_i\in V_D$ on $l$, the updated weight $w(\triangle vq_{i-1}q_i)$ will be the weight of the SST $\triangle vq_jq_i$, where $q_j$ is the last point on $l$ before $q_i$ that belongs to $V_D$. 

    Overall, the running time of both steps is linear in the number of auxiliary subsegments, which is $O(n^6)$ (as the size of $C_v$).
\end{proof}

We conclude that in each iteration of the algorithm, the total time for updating the weights of the SSTs for a single vertex $v$ of $P$ is $O(n^{6})$, and thus the total running time for all vertices is $O(n^{7})$.

\subsection{The heaviest path in the DAG}
Our goal now is to prove an equivalence between the heaviest path in $G$ and a largest restricted polygon in $P_v$ that contains $v$. Computing the heaviest path in a DAG $G$ can be done in $O(|V(G)|+|E(G)|)$ time by sorting the vertices in topological order and then running a dynamic programming algorithm (see, e.g., \cite{noltemeier1975algorithm}). Since $|V(G)|+|E(G)|=O(n^4)$, and constructing $G$ takes $O(n^6)$ time, we have an algorithm with running time of $O(n^6)$ for computing the largest restricted polygon in $P_v$.

Below, we show an equivalence between maximal paths in $G$ and restricted polygons in $P_v$ that contain $v$ as a vertex.

\begin{claim}\label{clm:max-path-endpoints}
    Any maximal path in $G$ starts with a directed subsegment $\overrightarrow{va}$ and ends with a directed subsegment $\overrightarrow{bc}$ such that $\overline{va},\overline{cv}$ lie on diagonal extensions, and the angle between $\overrightarrow{bc}$ and $\overrightarrow{cv}$ is at most $180^\circ$.
\end{claim}
\begin{proof}
    Consider a maximal directed path $\pi=s_1,s_2,\dots,s_k$ in $G$. Since $\pi$ is maximal, there is no directed subsegment $s_0$ such that there is an edge in $G$ between $s_0$ and $s_1$. Denote $s_1=\overrightarrow{xa}$, and assume by contradiction that $x\neq v$. If $x$ is not on $\partial P_v$, then clearly there is a directed subsegment $s_0=\overrightarrow{yx}$ on the same diagonal extension, and $\angle yxa=180^\circ$ so there is an edge between $s_0$ and $s_1$, a contradiction. Otherwise, $x$ is on $\partial P_v$. If $x$ is not a vertex of $P$, then by similar arguments we get a contradiction because any edge of $P_v$ lies on a diagonal extension. If $x$ is a convex vertex of $P$, then there must be a point $y$ from $V_D$ on one of its incident edges before $x$, and the angle between $\overrightarrow{yx}$ and $\overrightarrow{xa}$ is at most $180^\circ$, so again we get a contradiction. If $x$ is a reflex vertex, then consider the diagonal extension of $\overline{vx}$ and let $y$ be a point from $V_D$ on $\overline{vx}$ before $x$. If the angle between $\overrightarrow{yx}$ and $\overrightarrow{xa}$ is larger than $180^\circ$, then $a$ is not visible from $v$. Again we get that there is an edge between $s_0=\overrightarrow{yx}$ and $\overrightarrow{xa}$, a contradiction.     
    We conclude that $s_1=\overrightarrow{va}$ and $\overrightarrow{va}$ lies on a diagonal extension.
    
    Denote $s_k=\overrightarrow{bc}$. Notice that by symmetric arguments (reversing the directions), $\overrightarrow{cv}$ must lie on a diagonal extension, and the angle between $\overrightarrow{bc}$ and $\overrightarrow{cv}$ is at most $180^\circ$.
\end{proof}

For a maximal directed path $\pi=s_1,s_2,\dots,s_k$ in $G$, let $Q(\pi) = \bigcup_{s \in \pi} T_v(s)$.

\begin{claim} \label{clm:any_maximal_path}
    For any maximal directed path $\pi$, $Q(\pi)$ is a restricted polygon that has $v$ as one of its vertices.
\end{claim}
\begin{proof}
    By the definition of $G$ and \Cref{clm:max-path-endpoints}, since $v$ is a convex vertex, the directed subsegments of $\pi$ together with $v$ enclose a convex polygon whose vertices are from $V_D$ and whose edges are subsegments of diagonal extensions --- a restricted polygon.

    To see that $Q(\pi)$ is exactly the polygon enclosed by the directed subsegments of $\pi$, recall that each triangle $T_v(s_i)$ connects $v$ to two consecutive vertices of $V_D$ on $\pi$. In other words, the set $T_v(s_1),\dots,T_v(s_k)$ is a triangulation of $Q(\pi)$.
\end{proof}

\begin{claim} \label{clm:any_maximal_in_dag}
    For any maximal restricted polygon $Q$ that has $v$ as one of its vertices, there exists a maximal directed path $\pi$ such that $Q(\pi)=Q$.
\end{claim}
\begin{proof}
    Since $Q$ is a restricted polygon, each of its edges is a concatenation of subsegments in $D_v$ (possibly a single subsegment). Denote $s_1, \dots, s_k$ the chain of subsegments going counterclockwise from $v$ along the border of $Q$. Each $s_i$ has a vertex $u_i$ in $G$. Since $Q$ is a convex polygon, each angle between two consecutive subsegments $s_i, s_{i+1}$ is convex. By the construction of $G$, there is an edge $e_i = (u_i, u_{i+1})$ in $G$. Thus, the directed path $\pi = e_1, \dots, e_k$ exists in $G$. Notice that given a path $\sigma \in G$, $Q(\sigma)$ goes along the subsegments represented by the vertices of the path and then closes in a convex polygon using an extra edge connecting to $v$ itself (connecting the end point of the segment represented by the last vertex in $\sigma$ to the start point of the segment represented by the first vertex in $\sigma$). This is done since we wish for our $G$ to be acyclic, thus, we add the last edge in $Q(\sigma)$ artificially without representing the respective connection using an edge in $G$. In the case of $\pi$, this is exactly the polygon $Q$. That is, $Q(\pi) = Q$.
\end{proof}

We obtain the following lemma.
\begin{lemma}\label{obs:heavy-paths-convex}
    Let $Q$ be a restricted polygon, and $v$ a vertex in $P$. 
    Then $Q$ is a largest restricted polygon that contains $v$ if and only if there exists a heaviest path $\pi$ in $G$ such that $Q(\pi)=Q$. Moreover, the weight of $\pi$ is exactly $\sum_{f \in \U, f \subseteq Q} w(f)$.
\end{lemma}
\begin{proof}
    The first part of the lemma follows from \Cref{clm:any_maximal_path} and \Cref{clm:any_maximal_in_dag}.
    The weight of $\pi$ is $\sum_{s \in \pi} w(T_v(s))$, and since $Q(\pi)$ is a restricted polygon, for every face $f\in \U$ it either contains all the cells from $C_v$ in this face or does not contain any cell in this face. Since the sum of weights of the cells in a face $f$ is equal to $w(f)$, we get that the weight of $\pi$ is exactly $\sum_{f \in \U, f \subseteq Q} w(f)$.
\end{proof}

We conclude that we can iterate over all vertices of $P$, and for each vertex $v$ find a largest restricted polygon that contains $v$. Then, we can choose the restricted polygon $Q$ that maximizes $\sum_{f \in \U, f \subseteq Q} w(f)$, as required. Since in this section we have only considered convex vertices, we now show how to handle reflex vertices.

\subsection{Handling visibility polygons of reflex vertices}
If $v$ is a reflex vertex, we split $P_v$ into two sub‐polygons $P_{v_1}$ and $P_{v_2}$, separated by one of the diagonal extensions from $v$. Such a diagonal extension always exists by the extension of each of the edges of $P$ which are adjacent in $v$.
We then run the algorithm on each of them separately, and get two restricted polygons $Q_1,Q_2$.
The observation below follows from \Cref{obs:heavy-paths-convex}. 
\begin{observation}\label{obs:heavy-paths-reflex}
    Let $Q^*$ be a largest restricted polygon in $P_v$ that contains $v$, then $w(Q^*)=w((Q^*\cap P_{v_1})\cup(Q^*\cap P_{v_2}))$ and thus $w(Q_1\cup Q_2)\ge w(Q^*)$.
\end{observation}

To summarize, we build the DAGs for all the vertices of $P$. If a vertex $v$ is a reflex vertex, we first split it into two polygons as described above, and it will have two DAGs associated with it. Constructing all the DAGs takes $O(n^7)$ time.
To find at most two restricted polygons $Q_1,Q_2\in\R$ such that $w(Q_1\cup Q_2)\ge w(Q^*)$, where $Q^*$ is the restricted polygon that maximizes $w(Q^*)$, we find the heaviest path in each DAG (two DAGS, and two restricted polygons for a reflex vertex). Then we take the polygon $Q_1=Q(\pi)$ (or two polygons $Q_1,Q_2$) that correspond to the heaviest path $\pi$ (or paths), over all DAGs. By \Cref{obs:heavy-paths-convex} and \Cref{obs:heavy-paths-reflex}, we obtain the following theorem.

\begin{theorem}[Meshed potato peeling]
    Given the set of faces $\U$, with a binary function $w:\U\rightarrow\{0,1\}$, one can compute in $O(n^7)$ time at most two restricted polygons $Q_1,Q_2\in\R$ such that $w(Q_1\cup Q_2)\ge w(Q^*)$, where $Q^*$ is the restricted polygon that maximizes $w(Q^*)$.
\end{theorem}

\section{More (rotten) potatoes}\label{sec:more-potatoes}
The rotten meshed potato peeling problem that we solve in \Cref{sec:potato} is rather restricted. However, in this section we show that the algorithm that we develop for this problem can be utilized for approximating a more natural version of the rotten potato peeling problem, defined as follows.

\begin{problem}[Rotten Potato Peeling]\label{prb:rotten-potatoes}
    Given a polygon $P$ with $n$ vertices, and a set $R$ of polygons in $P$ with a total of $k$ vertices, find a convex polygon $Q$ in $P$ that maximizes the area of $Q \setminus R$.  
\end{problem}

Below, we present a $\frac14$-approximation algorithm for the rotten potato peeling problem, that runs in $O(kn^5)$ time, using the same general approach as in \Cref{sec:potato}.

First, observe that we can assume that an optimal solution $Q^*$ to \Cref{prb:rotten-potatoes} is a maximal polygon.
We prove the following theorem.
\begin{theorem}\label{thm:4-approx}
    For any maximal convex polygon $Q$ in $P$, there exists a restricted polygon $Q'$ that contains a vertex $v$ of $P$, such that $\text{Area}(Q' \setminus R)\ge \frac14 \text{Area}(Q\setminus R)$. 
    Moreover, if $v$ is a reflex vertex, split $P$ into two polygons $P_1,P_2$ using one of the diagonal extensions from the edges incident to $v$. Then $Q'$ is contained in one of $P_1,P_2$.
\end{theorem}
\begin{proof}
Since $Q$ is maximal, it must contain a vertex $v$ of $P$ (otherwise it can be expanded as in the proof of \Cref{thm:approx}). Moreover, $Q$ must lie entirely within the visibility polygon of $v$ (because $Q$ is convex).

If $v$ is a convex vertex, then by \Cref{thm:approx}, $Q$ can be covered by at most two restricted polygons $Q_1,Q_2$, and by the construction, both of them have $v$ as a vertex. Clearly, one of $Q_1,Q_2$ has $\text{Area}(Q_i \setminus R)\ge \frac12 \text{Area}(Q\setminus R)$. 

Else, $v$ is a reflex vertex. Split $P$ into two polygons $P_1,P_2$ using one of the diagonal extensions from the edges incident to $v$. Denote $Q_1=Q\cap P_1$, and $Q_2=Q\cap P_2$. Again, one of $Q_1,Q_2$ has $\text{Area}(Q_i \setminus R)\ge \frac12 \text{Area}(Q\setminus R)$. Assume it is $Q_1$.
By applying \Cref{thm:approx} on $P_1$ and $Q_1$, we get that there exists two restricted polygons $Q_3,Q_4$ in $P_1$, both of them have $v$ as a vertex, that together cover $Q_1$ (because in $P_1$, $Q_1$ contains the convex vertex $v$). We get that $\max\{\text{Area}(Q_3 \setminus R), \text{Area}(Q_4 \setminus R)\}\ge \frac12 \text{Area}(Q_1\setminus R)\ge \frac14 \text{Area}(Q\setminus R)$, as required.
\end{proof}

We construct a DAG for each convex vertex $v$ of $P$, as we did in \Cref{sec:potato}, except for the weight function $w$. Instead, we define the weight of each vertex (directed subsegment) $s$ of the DAG as $w'(s)=\text{Area}(T_v(s) \setminus R)$, i.e. the weight of $s$ is the area of the ``good'' regions that $T_v(s)$ covers. We compute the weight of each vertex by simply computing the intersection of the corresponding triangle with $R$, which takes $O(|R|)=O(k)$ time per triangle, and $O(kn^5)$ in total, as there are $O(n^5)$ SSTs in all the DAGs. 

For each reflex vertex $v$ of $P$, we split $P$ into two polygons $P_1,P_2$ using one of the diagonal extensions from the edges incident to $v$, as in \Cref{thm:4-approx}. We then construct a DAG for each of $P_1,P_2$ separately, as if it was a visibility polygon of a convex vertex.

To find the maximal polygon, we look for maximal paths in all DAGs that we have constructed. 
For a path $\pi$ in the graph, we define $Q(\pi)$ as before, and since the SSTs that correspond to the vertices of $\pi$ are a triangulation of $Q(\pi)$, we get that area of the ``good'' regions that $Q(\pi)$ covers is exactly the sum of weights of these triangles. By similar arguments as in \Cref{sec:potato}, we get that the heaviest path in the graph is equivalent to a restricted polygon that covers the most area from $P/R$. We then take the maximal solution over all the vertices of $P$. 
We conclude that this algorithm finds in $O(kn^5)$ a restricted polygon $Q$ that has the properties from \Cref{thm:4-approx}: it contains a vertex $v$ of $P$, and if $v$ is a reflex vertex, it is contained in one of the subpolygons obtained by splitting $P$ using one of the diagonal extensions from the edges incident to $v$. Since $Q$ is the maximal polygon that has these properties, we obtain the following theorem.

\begin{theorem}
    Given a polygon $P$ with $n$ vertices, and a set $R$ of polygons in $P$ with a total of $k$ vertices, there is an algorithm that in $O(kn^5)$ time finds a convex polygon $Q$ in $P$ such that $\text{Area}(Q \setminus R)\ge \frac14 \text{Area}(Q^* \setminus R)$, where $Q^*$ is a convex polygon that maximizes $\text{Area}(Q^* \setminus R)$.
\end{theorem}

\section{Conclusion}
We have presented a new algorithm for the minimum convex cover problem that retains the $O(\log n)$ approximation ratio of Eidenbenz and Widmayer \cite{EidenbenzWidmayer03}, while dramatically improving the running time. By removing the dependence on second-order grid triangulations, and using a DAG-based approach to efficiently find maximum area convex polygons whose vertices are constrained to $P$ or the first-order grid ($\A_D$), we achieve a worst-case running time of $O(n^8)$. This is a substantial improvement over the $O(n^{29} \log n)$ complexity of prior work.

While $O(n^8)$ is still impractical, our method suggests that further improvements may be possible, and may also lead to advances in other visibility and covering problems.

\section*{Acknowledgments}
This research was supported by the ISRAEL SCIENCE FOUNDATION (grant No. 2135/24).

\bibliographystyle{abbrv}
\bibliography{refs} 

\end{document}